\documentclass[
final
]{dmtcs-episciences}

\usepackage[utf8]{inputenc}
\usepackage{subfigure}

\usepackage{amsmath}
\usepackage{amsfonts}
\usepackage{amssymb}
\usepackage{amsthm}

\newtheorem{fact}{Fact}
\newtheorem{theorem}{Theorem}
\newtheorem{lemma}{Lemma}
\newtheorem{corollary}{Corollary}

\newcommand{\mypar}[1]{ \left( #1 \right) }
\newcommand{\mymatrix}[2]{ \mypar{ \begin{array}{#1} #2 \end{array} } }
\newcommand{\myvec}[1]{ \mymatrix{c}{#1} }

\newcommand{\ket}[1]{|#1\rangle}
\newcommand{\bra}[1]{\langle #1|}

\author{Aleksejs Naumovs\affiliationmark{1}\thanks{Naumovs gave the constructions in Sections 4 and 8 when he was respectively a 11-th and 12-th grade high school student in Riga Secondary School No.13, R\={\i}ga, Latvia.}
  \and Maksims Dimitrijevs\affiliationmark{2}\thanks{Dimitrijevs was partially supported by University of Latvia projects AAP2016/B032 ``Innovative information technologies'' and ZD2018/20546 ``For development of scientific activity of Faculty of Computing''.}
  \and Abuzer Yakary{\i}lmaz\affiliationmark{2}\thanks{Yakary\i lmaz was partially supported by Akad\={e}misk\={a} person\={a}la atjaunotne un kompeten\v{c}u pilnveide Latvijas Universit\={a}t\={e} l\={\i}guma Nr. 8.2.2.0/18/A/010 LU re\c{g}istr\={a}cijas Nr. ESS2018/289, ERC Advanced Grant MQC, the ERDF project 1.1.1.5/18/A/020, and Kvantu datori ar konstantu atmi\c{n}u l\={\i}g. Nr. 1.1.1.15/19/A/005 LU re\c{g}istr\={a}cijas Nr. ESS2020/338.}
  }
\title[The minimal PFAs and QFAs recognizing uncountably many languages with fixed cutpoints]{The minimal probabilistic and quantum finite automata recognizing uncountably many languages with fixed cutpoints}
\affiliation{
  Faculty of Physics, Mathematics and Optometry, University of Latvia
  \\
  Center for Quantum Computer Science, Faculty of Computing, 
University of Latvia
}
\received{2019-5-13}
\revised{2020-4-4}
\accepted{2020-4-8}
\begin{document}
\publicationdetails{22}{2020}{1}{13}{5450}
\maketitle
\begin{abstract}
  It is known that 2-state binary and 3-state unary probabilistic finite automata and 2-state unary quantum finite automata recognize uncountably many languages with cutpoints. These results have been obtained by associating each recognized language with a cutpoint and then by using the fact that there are uncountably many cutpoints. In this note, we prove the same results for fixed cutpoints: each recognized language is associated with an automaton (i.e., algorithm), and the proofs use the fact that there are uncountably many automata. For each case, we present a new construction.
\end{abstract}
\section{Introduction}

It is a well-known fact that all \textit{Turing machines} (TMs) form a countable set as each TM has a finite description. Moreover, since each TM defines (recognizes) a single language as a recognizer, the class of languages recognized by TMs, called \textit{recursively enumerable languages}, also forms a countable set. On the other hand, all (unary or binary) languages form an uncountable set. Thus, we can easily conclude that there are some languages that cannot be recognized by (associated with) any TM \cite{Tur37}.

As a very restricted form of TM, a \textit{finite state automaton} (FSA) \cite{RS59} reads the input once from left to right and then gives its decision. Their computational power is significantly less, and the class of languages recognized by them is called the class of \textit{regular languages}, a proper sub-class of recursively enumerable languages. On the other hand, a FSA can be enhanced by making probabilistic choices, called \textit{probabilistic finite automaton} (PFA) \cite{Rab63}. In contrast to FSAs or TMs, all PFAs form an uncountable set if they are allowed to use real-valued transition probabilities. Then, one may ask whether the languages recognized by PFAs also form an uncountable set or not. 

A FSA can either accept or reject a given input string, and so, it is easy to classify the set of all strings into two sets, i.e., the language recognized by the automaton and its complement. However, a PFA defines a probability distribution on the input strings. Therefore, in order to split the set of all strings into two sets, we additionally use a threshold, called cutpoint. That is, a PFA defines (recognizes) a language with a cutpoint, which contains all strings accepted with probability greater than the cutpoint. Such a language is called \textit{stochastic} \cite{Rab63}. 

A single PFA may define different stochastic languages with different cutpoints. Rabin, in his seminal paper on PFAs \cite{Rab63}, presented a two-state PFA using only rational-valued transitions that recognizes uncountably many languages with cutpoints. That is, even a single PFA can define an uncountable set of languages. Since there is a single PFA, the uncountability result is based on the fact that there are uncountably many cutpoints. 

A similar result was given for a quantum counterpart of PFA, \textit{quantum finite automaton} (QFA), \cite{ShurY14} that a  two-state rational-valued QFA over a single letter  (unary) alphabet can define an uncountable set of languages. This result is stronger in a way that Rabin's result was given for binary languages but the quantum result was given for unary languages. Besides, a similar uncountability result for unary PFAs can be obtained only for three states \cite{ShurY16}.

The minimal binary and unary PFAs and unary QFAs defining uncountably many languages with cutpoints have two, three, and two states, respectively. All these three results were given by using the fact that the cardinality of cutpoints is uncountable. We find it interesting and natural to obtain the same results by using the fact that the cardinality of automata (i.e., algorithms) is uncountable. In other words, we obtain the same result for fixed cutpoints. 

In this note, we show that two-state binary PFAs, two-state unary QFAs, and three-state unary PFAs can recognize uncountably many languages with fixed cutpoints. 

We present the notations and definitions used throughout this note in the next section. Then, we explain the original proof given by Rabin~\cite{Rab63} in Section~\ref{sec:Rabin}, and we present our modification on Rabin's proof in Section~\ref{sec:our-rabin}. In Section~\ref{sec:quantum}, we explain the quantum version of Rabin's proof. After this, we present our quantum result in Section~\ref{sec:our-quantum}. Lastly, we explain the known uncountability results for unary PFAs in Section~\ref{sec:unary-PFA}, and we present our construction for unary PFAs in Section~\ref{sec:our-unary-PFA}. 

\section{Background}

\textit{The input alphabet} is represented by $ \Sigma $. \textit{The empty string} is represented by $ \varepsilon $. The set of all strings defined on $ \Sigma $ is denoted by $\Sigma^*$. Moreover, $ \Sigma^+ = \Sigma^* \setminus \{ \varepsilon \} $. For any given string $ x \in \Sigma^* $, $ |x| $ represents its length, $ x^r $ represents its reverse string, and, if $|x|>0$, $ x[i] $ denotes its $i$-th symbol, where $ 1 \leq i \leq |x| $. For any binary string $ x \in \{0,1\}^+ $, $ bin(x)  $ represents  its binary representation, i.e., $ bin(x) = 0.~x[1] ~ x[2] ~ \cdots ~ x[|x|] $.  The accepting probability of an automaton $ M $ on a given input string $ x $ is given by $ f_M(x) $. The cutpoints are defined in the interval $ [0,1) $. 

The (probabilistic or quantum) state of an $n$-state automaton is represented by an $n$-dimensional (stochastic or unit) column vector. Any given input is read from left to right and symbol by symbol. The computation starts in a single state. The computation is traced by an $n$-dimensional column vector. For each symbol, an operator is associated. Whenever this symbol is read, the current vector is multiplied from the left by the corresponding linear operator (represented as $(n \times n)$-dimensional matrix). Based on the final state vector, the accepting probability of the input by the automaton is calculated. 

For pre- and post-processing, probabilistic and quantum finite automata (PFAs and QFAs, respectively) can read one specific symbol before reading the input and one specific symbol after reading the input, respectively. However, in this paper, the models do not do any pre- or post-processing. 

Formally, an $n$-state ($n>0$) PFA $ P $ is a 5-tuple
\[
    P = \{ S, \Sigma, \{ A_\sigma \mid \sigma \in \Sigma \}, s_i,S_a \},
\]
where $ S = \{ s_1,\ldots,s_n \} $ is the set of states, $ A_\sigma $ is an $(n\times n)$-dimensional (left) stochastic matrix associated to symbol $ \sigma \in \Sigma $, $ s_i \in S $ is the initial state, and $ S_a \subseteq S $ is the set of accepting state(s).

The initial probabilistic state $ v_0 $ is a zero-one stochastic vector having 1 in its $ i $-th entry. For the empty string, the final probabilistic state is $ v_f = v_0 $. For a given input $ x \in \Sigma^+ $, the final probabilistic state is
\[
    v_f = A_{x[|x|]} A_{x[|x|-1]} \cdots A_{x[1]} v_0.
\]
Then, the accepting probability of $ x $ by $ P $ is 
\[
    f_P(x) = \sum_{s_j \in S_a} v_f(j).
\]

Any language that is recognized by a PFA with cutpoint is called stochastic, and the class of stochastic languages forms an uncountable set \cite{Rab63}.

We assume the reader is familiar with the basic concepts (i.e., unitary evolution and projective measurements) used in quantum computation (see \cite{NC00} and \cite{SayY14} for a complete and quick reference, respectively). As a convention of quantum computation, a (complex-valued) column vector $ v $ can be represented as $ \ket{v} $ (and its conjugate transpose row vector is represented as $ \bra{v} $). 

In the literature, there are many different definitions of QFAs \cite{AY15A}. Here we use the known most restricted version, so-called Moore-Crutchfield or Measure-Once QFA \cite{MC00}.

Formally, an $n$-state ($n>0$) QFA $ M $ is a 5-tuple
\[
    M = \{ Q, \Sigma, \{ U_\sigma \mid \sigma \in \Sigma \}, q_i,Q_a \},
\]
where $ Q = \{ q_1,\ldots,q_n \} $ is the set of states, $ U_\sigma $ is an $(n\times n)$-dimensional unitary matrix associated to symbol $ \sigma \in \Sigma $, $ q_i \in Q $ is the initial state, and $ Q_a \subseteq Q $ is the set of accepting state(s). The vector $ \ket{q_j} $ represents the $n$-dimensional zero-one unit vector having 1 in its $j$-th entry, where $ 1 \leq j \leq n$.

The initial quantum state is $ \ket{v_0} = \ket{q_i} $. For the empty string, the final quantum state is $ \ket{v_f} = \ket{v_0} $. For a given input $ x \in \Sigma^+ $, the final quantum state is
\[
    \ket{v_f} = U_{x[|x|]} U_{x[|x|-1]} \cdots U_{x[1]} \ket{v_0}.
\]
Then, the accepting probability of $ x $ by $ M $ is 
\[
    f_M(x) = \sum_{q_j \in Q_a} \big| \ket{v_f}(j) \big|^2,
\]
which is obtained by making a measurement in the computational basis at the end of the computation. 

Any language recognized by a QFA (even if it is the most general variant of QFA) with cutpoint was shown to be stochastic \cite{AY15A}, and vice versa in most of the cases (any stochastic language can be recognized by almost all variants of QFAs \cite{YS11,AY15A}). 

A single-letter alphabet and any automaton defined over it is called unary. Similarly, any two-letter alphabet and any automaton defined over it is called binary.

\section{Rabin's proof}
\label{sec:Rabin}

We start with the proof given by Rabin. We name the PFA presented by Rabin as $ P $, which has two states, say $ s_1 $ and $ s_2 $. The computation starts in $ s_1 $ and the only accepting state is $ s_2 $. The PFA $ P $ operates on binary strings, defined on $ \Sigma = \{0,1\} $. Let $ x \in \Sigma^+ $ be a given input. (It is clear that $ f_P(\varepsilon) = 0 $.)

  We trace the computation of $ P $ by a 2-dimensional column stochastic vector (probabilistic state), which is $ v_0 = \myvec{1 \\ 0} $ at the beginning. After reading symbol $ 0 $ and $ 1 $, $ P $ applies the following stochastic operators to its probabilistic state
\[
	A_0 = \mymatrix{cc}{1 & 1/2 \\ 0 & 1/2}    
	\mbox{ and } 
	A_1 = \mymatrix{cc}{1/2 & 0 \\ 1/2 & 1},
\] 
respectively. 

The probabilistic state after reading $ x $ is $ v_f = v_{|x|} $. 
%\[
%	v_{|x|} = A_{x[|x|]} \cdot A_{x[|x|-1]} \cdot \cdots \cdot A_{x[1]} v_0.
%\]
Then, $ f_P(x) $ is the second entry of $ v_{|x|} $. By using induction, we can easily see that $ f_P(x) = bin(x^r) $. Thus, for any rational number between 0 and 1, say $t \in (0,1) \cap \mathbb{Q} $, there exists at least one binary string $ y $ such that $ f_P(y) = t $. 

Let $ \lambda_1 < \lambda_2 $ be two real-valued cutpoints between 0 and 1. Since the rational numbers are dense on real numbers, there exists at least one binary string, say $ z $, such that $ \lambda_1 <  f_{P}(z) < \lambda_2 $. Thus we can conclude that $ P $ with cutpoint $ \lambda_1 $ recognizes the language that is a superset of the language recognized by $ P $ with cutpoint $ \lambda_2 $. More generally, for any given cutpoint,  $ P $ recognizes a different language. Since there are uncountably many cutpoints, $ P $ recognizes uncountably many stochastic languages (, and hence the class of stochastic languages forms an uncountable set).

In this proof, the existence of uncountably many stochastic languages is shown based on a single PFA, and thus the result follows from the fact that there are uncountably many cutpoints. We find it interesting and natural to obtain the same result by using the fact that there are uncountably many PFAs. In other words, we fix the cutpoint and show that there is a set of uncountably many PFAs such that each PFA recognizes a different language with this fixed cutpoint. 

We note that the result for 3-state PFAs is trivial since it is known that (e.g., see \cite{Paz71}) if language $ L $ is defined by an $n$-state PFA $ N $ with cutpoint $ \lambda \in [0,1] $, then for any given $ \lambda' \in (0,1) $, there always exists an $(n+1)$-state PFA, say $ N' $, recognizing $ L $ with cutpoint $ \lambda'  $.

\section{Our modification on Rabin's proof}
\label{sec:our-rabin}

Here we show that, for any given nonzero cutpoint, two-state PFAs can define uncountably many languages. We believe that our proof is more elegant since it is based on the existence of uncountably many PFAs (algorithms).

\begin{lemma}
	For a given $ \alpha \in (0,1) $, there exists a 2-state PFA $ P_\alpha $ accepting any non-empty binary string $ x $ with probability $ \alpha \cdot bin(x^r) $. 
\end{lemma}
\begin{proof}
	We use the PFA given by Rabin after modifying the transition matrix for symbol $1$. The PFA $ P_\alpha $ has two states: the first one is the initial and the only accepting state is the second one. The transition matrices for symbols 0 and 1 are	
	\[
		A_{\alpha,0} = \mymatrix{cc}{1 & 1/2 \\ 0 & 1/2} 
		\mbox{ and }
		A_{\alpha,1} = \mymatrix{cc}{1-\alpha/2 & (1-\alpha)/2 \\ \alpha/2 & (1+\alpha)/2},
	\]
	respectively. 
	We use induction to prove our Lemma.
	
	\textit{Base case:} After reading string $ x=0 $ or $x=1$, the probabilistic state is
	\[
		v_0 = \myvec{1\\0} = A_{\alpha,0} \myvec{1\\0} 
		\mbox{ or }
		v_1 = \myvec{1 - \alpha/2  \\ \alpha/2} = A_{\alpha,1} \myvec{1\\0}, 
	\]
	respectively, and hence the accepting probability for string 0 or 1 is 0 or $\alpha/2 = \alpha \cdot bin(1^r) $, respectively. $ \vartriangleleft $
	
	\textit{Inductive step:} Assume that, after reading $ x \in \{0,1\}^+ $, the probabilistic state is
	\[
		v_{|x|} = \myvec{ 1- \alpha \cdot bin(x^r) \\ \alpha \cdot bin(x^r) }.
	\]
	Then, after reading $ x0 $, the new probabilistic state is
	\begin{eqnarray*}
		v_{|x|+1} &  = & \mymatrix{cc}{1 & 1/2 \\ 0 & 1/2}  \myvec{ 1- \alpha \cdot bin(x^r) \vspace*{3px} \\ \alpha \cdot bin(x^r) } 
		\\
		 & = & 
		\myvec{ 1 - \alpha \cdot \frac{bin(x^r)}{2} \vspace*{7px} \\  \alpha \cdot \frac{bin(x^r)}{2}  }
		\\
		& = &
		\myvec{ 1 - \alpha \cdot bin((x0)^r) \vspace*{3px} \\ \alpha \cdot bin((x0)^r) }
	\end{eqnarray*}
	and hence the accepting probability is $\alpha \cdot bin((x0)^r)$.
Similarly, after reading $ x1 $, the new probabilistic state is
	\begin{eqnarray*}
		v_{|x|+1} &  = & \mymatrix{cc}{1-\alpha/2 & (1-\alpha)/2 \\ \alpha/2 & (1+\alpha)/2}  \myvec{ 1- \alpha \cdot bin(x^r) \vspace*{3px} \\ \alpha \cdot bin(x^r) } 
		\\
		 & = & 
		\myvec{ \overline{1} \vspace*{5px} \\  \dfrac{\alpha}{2} - \dfrac{\alpha^2 \cdot bin(x^r)}{2} + \dfrac{\alpha \cdot bin(x^r)}{2} + \dfrac{\alpha^2 \cdot bin(x^r)}{2}  }
		\\
		& = &
		\myvec{ \overline{1} \vspace*{3px} \\ \alpha \cdot \mypar{ \dfrac{1}{2} + \dfrac{bin(x^r)}{2} } }
		\\
		& = &
		\myvec{ \overline{1} \vspace*{3px} \\ \alpha \cdot bin((x1)^r)  },
	\end{eqnarray*}
	and hence the accepting probability is $\alpha \cdot bin((x1)^r)$, where $ \overline{1} $ is the value making the column summation to 1.	$ \vartriangleleft $
	
	We conclude that for any given input string $ x \in \{0,1\}^+ $, $ f_{P_\alpha}(x) = \alpha \cdot bin (x^r) $.
\end{proof}

\begin{theorem}
	For any given cutpoint $ \lambda \in (0,1) $, 2-state PFAs recognize uncountably many stochastic languages with cutpoint $ \lambda $.
\end{theorem}
\begin{proof}

Let $ 0< \alpha_1 < \alpha_2 < 1 $ be two real numbers. Then, there exists a string $ z $ such that $ \alpha_1 < f_P(z) < \alpha_2  $, where $ P $ is the PFA given by Rabin. Then, we can easily derive these two inequalities
	\[
		\lambda < \frac{\lambda}{\alpha_1} f_P(z) 
		\mbox{ and }
		 \frac{\lambda}{\alpha_2} f_P(z) < \lambda. 
	\]
	Due to Lemma 1, we can conclude that
	\[
		\lambda < f_{P_{\frac{\lambda}{\alpha_1}}} (z) \mbox{  and }
		 f_{P_{\frac{\lambda}{\alpha_2}}} (z) < \lambda.
	\]
	Thus, the string $ z $ is in the language recognized by $ P_{\frac{\lambda}{\alpha_1}} $ with cutpoint $ \lambda $, and, it is not in the language recognized by $ P_{\frac{\lambda}{\alpha_2}} $ with cutpoint $ \lambda $. Therefore, for any given two different real numbers between 0 and 1, there exist two PFAs such that they recognize different languages with cutpoint $ \lambda $. 
\end{proof}

\begin{corollary}
	2-state PFAs can recognize uncountably many stochastic languages with cutpoint $1/2$.
\end{corollary}

\section{Quantum version of Rabin's result}
\label{sec:quantum}

The quantum version of Rabin's result was given for 2-state unary real-valued QFAs \cite{ShurY14}, defined as $ M_\alpha = \{ \{q_1,q_2\},\{0\},R_{\alpha \cdot 2\pi},q_1,\{q_1\} \} $, where $ \alpha \in (0,1) $ is an irrational number and $ R_{\alpha \cdot 2\pi} $ is the counter-clockwise rotation with angle $ \alpha \cdot 2\pi $ on $ \ket{q_1}-\ket{q_2} $ plane. It is clear that $ R_{\alpha \cdot 2\pi} $ is a real-valued unitary matrix. Moreover, all quantum states on $ \ket{q_1}-\ket{q_2} $ plane form the unit circle.

We fix $ \alpha $ such that $ R_{\alpha \cdot 2\pi} = \mymatrix{cr}{3/5&-4/5 \\ 4/5 & 3/5 }$. The automaton $ M_\alpha $ starts in $ \ket{q_1} = \myvec{ 1 \\ 0} $, and rotates on $ \ket{q_1}-\ket{q_2} $ plane with angle $ \alpha \cdot 2\pi $ for each input symbol $ 0 $. Let $ \ket{v_j} $ be the quantum state after reading $j$ symbols, i.e., $ \ket{v_j} = \cos(j\cdot \alpha \cdot 2\pi) \ket{q_1} + \sin(j\cdot \alpha \cdot 2\pi) \ket{q_2} $. Then, the accepting probability of $ 0^j $ is 
\[
    f_{M_\alpha}(0^j) = \cos^2(j\cdot \alpha \cdot 2\pi).
\]

Since $ \alpha $ is irrational, $ \{ \ket{v_j} \mid j >0 \} $, the set of all quantum states that $ M_\alpha $ can be in, is dense on the unit circle. Similarly, $ \{ \cos^2(j\cdot \alpha \cdot 2\pi) \mid j \geq 0 \} $, the accepting probabilities of all inputs by $ M_\alpha $, is dense on $ (0,1) $. Therefore, for any given two cutpoints $ \lambda_1 < \lambda_2 $, there is always an input $ 0^j $ such that 
\[
    \lambda_1 <  \cos^2(j\cdot \alpha \cdot 2\pi) < \lambda_2.
\]
Therefore, the automaton $ M_\alpha $ recognizes a different language for each cutpoint in $ (0,1) $. In other words, the class of languages recognized by $ M_\alpha $ with cutpoints forms an uncountable set.

\section{Our result for unary QFAs}
\label{sec:our-quantum}

We use the same automaton family $ \{ M_\alpha \} $ given in the previous section by restricting the irrational parameter $ \alpha \in (0,1/4) $.

Let $ \alpha $ and $ \beta $ be two different irrational numbers in $ (0,1/4) $. Then, their digit by digit binary representations are as follows:
\[
    \alpha = 0.\alpha_1 \alpha_2 \alpha_3 \alpha_4 \cdots \alpha_j \cdots
\]
and
\[
    \beta = 0.\beta_1 \beta_2 \beta_3 \beta_4 \cdots \beta_j \cdots ,
\]
where $ \alpha_1 = \alpha_2 = \beta_1 = \beta_2 = 0 $.

Since $ \alpha $ and $ \beta $ are different, there exists a minimal $ j > 2 $ such that $ \alpha_j \neq \beta_j $. 

Suppose that $ \alpha_j = 1 $ and $\beta_j=0$. We use the input of length $ 2^{j-3} $, say $ x_j $. After reading $ x_j $, $ M_\alpha $ and $ M_\beta $ rotate by angles 
\[
    \theta_1 = 2^{j-3} \cdot \alpha \cdot 2 \pi 
    ~~\mbox{and}~~
     \theta_2 = 2^{j-3} \cdot \beta \cdot 2 \pi .
\]
The angles $ \theta_1 $ and $ \theta_2 $ are congruent to 
\[
    \label{eq:congruent}
    \overline{\theta_1} =
    (0. \alpha_{j-2}\alpha_{j-1} 1) 2\pi + \theta_1'
    ~~\mbox{and}~~
    \overline{\theta_2} =
    (0. \alpha_{j-2}\alpha_{j-1} 0) 2\pi + \theta_2'
\]
modulo $ 2\pi $, respectively, where $ \theta_1',\theta_2' < \frac{\pi}{4} $. 
We can rewrite $ \overline{\theta_1} $ and $\overline{\theta_2}$ as
\[
   \overline{\theta_1}=
   \alpha_{j-2} \cdot \pi +
    \alpha_{j-1} \cdot \frac{\pi}{2} +
    \frac{\pi}{4} + \theta_1'
    ~~\mbox{and}~~
    \overline{\theta_2} =
    \alpha_{j-2} \cdot \pi +
    \alpha_{j-1} \cdot \frac{\pi}{2} +
     \theta_2'.
\]
The quantum states of $ M_\alpha $ and $ M_\beta $ lie in the same quadrant after reading the input $ x_j $. Here the values of $ \alpha_{j-2} $ and $ \alpha_{j-1} $ determine the number of a quadrant. But, in any case, we can have either 
\[
    f_{M_\alpha}(x_j) < \frac{1}{2} < f_{M_\beta}(x_j)
\]
or
\[
    f_{M_\beta}(x_j) < \frac{1}{2} < f_{M_\alpha}(x_j).
\]
In the 1st and 3rd quadrants, we have $ \cos^2(\overline{\theta_1}) < \frac{1}{2} $ and $ \cos^2(\overline{\theta_2}) > \frac{1}{2} $ as the quantum state of $ M_\alpha $ is closer to the $ \ket{q_2} $-axis and the quantum state of $ M_\beta $ is closer to the $ \ket{q_1} $-axis. In the 2nd and 4th quadrants,  
we have $ \cos^2(\overline{\theta_1}) > \frac{1}{2} $ and $ \cos^2(\overline{\theta_2}) < \frac{1}{2} $ as the quantum state of $ M_\alpha $ is closer to the $ \ket{q_1} $-axis and the quantum state of $ M_\beta $ is closer to the $ \ket{q_2} $-axis. We listed all cases in the following table.
\[
    \begin{array}{|c|c|c|c|c|}
    \hline 
        \alpha_{j-2} & \alpha_{j-1} & quadrant &  f_{M_\alpha}(x_j)  & f_{M_\beta}(x_j) 
        \\ \hline \hline
         0 & 0 & I & < 1/2 & > 1/2
        \\ \hline 
        0 & 1 & II & > 1/2 & <1/2
        \\ \hline 
        1 & 0 & III & < 1/2 & > 1/2
        \\ \hline 
        1 & 1 & IV & > 1/2 & < 1/2
        \\ \hline 
    \end{array}
\]

The case in which $ \alpha_j =0 $ and $ \beta_j =1 $ is symmetric. By using the same arguments, we obtain the above table after interchanging the last two headers ($ f_{M_\alpha}(x_j) $ and $ f_{M_\beta}(x_j) $). Therefore, we can conclude the following result.

\begin{theorem}
    For any given two irrational numbers $ \alpha $ and $ \beta $ in $ (0,1/4) $, the QFAs $ M_\alpha $ and $ M_\beta $ recognize different languages with cutpoint $ \frac{1}{2} $.
\end{theorem}

\begin{corollary}
    The class of languages recognized by 2-state unary real-valued QFAs with cutpoint $ \frac{1}{2} $ forms an uncountable set.
\end{corollary}

\section{Unary PFAs}
\label{sec:unary-PFA}

Rabin's proof was given for binary 2-state PFAs. Unary 2-state PFAs can recognize only few regular languages with cutpoints \cite{Paz71,ShurY16}. Besides, any unary $n$-state PFA can recognize at most $ n $ nonregular languages with cutpoints \cite{Paz71}. Therefore, there is no direct counterpart of Rabin's result for unary PFAs. However, we can still show that unary PFAs can define uncountably many stochastic languages. The proof was given first for a family of 4-state unary PFAs \cite{ShurY14}, and then for a family of 3-state unary PFAs \cite{ShurY16}.

The former result was already given for a fixed cutpoint ($\frac{1}{4}$) (by combining the quantum result given in Section \ref{sec:quantum} and Turakainen conversion technique \cite{Tur69,Tur75}). The latter result was given for the pairs of PFAs and cutpoints, i.e., $ \{ (Q_x,\lambda_x) \mid x \in (0,1/2] \} $, and hence the proof is still based on the cardinality of cutpoints. 

In this section, we give the details of the latter result, and then, in the next section, we present our construction for the fixed cutpoint.

For each $ x \in (0,1/2] $, $ Q_x $ is a 3-state unary PFA over the alphabet $ \{0\} $. The first state is the starting state and the last state is the only accepting state. The single transition matrix for symbol $0$ is
\[
    B_x = \mymatrix{ccc}{0 & 0 & x \\ 1 & 0 & x \\ 0 & 1 & 1-2x}. 
\]
The eigenvalues of $ B_x $ are
\[
    r_1 = 1, ~~~ r_2 = -x + \sqrt{x-x^2} \cdot i, ~\mbox{ and }~ r_3=-x - \sqrt{x-x^2} \cdot i, 
\]
where $r_2$ and $ r_3 $ are complex conjugates of each other. 

For the input $ 0^m $, the accepting probability $ Q_x $ is calculated as 
\[
    f_{Q_x}(0^m) = \myvec{0~~0~~1} B_x^m \myvec{1\\0\\0}.
\]
In other words, $ f_{Q_x}(0^m) $ is equal to $ B_x^m (3,1) $. 

Due to the Cayley-Hamilton theorem and simplicity of eigenvalues, each entry of $ B_x^m  $ is a linear combination of $ r_1^m $, $ r_2^m $, and $ r_3^m $. Thus, we can write  $ f_{Q_x}(0^m) $ as
\[
    a \cdot 1^m + (b+c\cdot i)\mypar{  -x + \sqrt{x-x^2} \cdot i }^m + (b-c\cdot i)\mypar{  -x - \sqrt{x-x^2} \cdot i }^m,
\]
where the coefficients of $ r_2^m $ and $ r_3^m $ are also conjugates of each other. By using the following initial conditions
\[
     f_{Q_x}(0^0) = 0, ~~  f_{Q_x}(0^1) = 0, ~\mbox{ and, }  f_{Q_x}(0^2) = 1,   
\]
the coefficients can be calculated as
\[
    a = \dfrac{1}{3x+1}, ~~~ b = -\dfrac{1}{6x+2}, ~ \mbox{ and } ~ c = \dfrac{x+1}{(6x+2)\sqrt{x-x^2}}.
\]

The polar forms of $ r_2 $ and $ r_3 $ are respectively
\[
    \sqrt{x} \mypar{ \cos \theta_x + i \cdot \sin \theta_x } ~ \mbox{ and } ~ \sqrt{x} \mypar{ \cos \theta_x - i \cdot \sin \theta_x },
\]
where $ \theta_x = \arccos\mypar{-\sqrt{x}} $. The polar forms of $ b+c\cdot i $ and $ b-c\cdot i $ are respectively
\[
    \sqrt{b^2+c^2} \mypar{ \cos \gamma_x + i \cdot \sin \gamma_x } ~ \mbox{ and } ~ \sqrt{b^2+c^2} \mypar{ \cos \gamma_x - i \cdot \sin \gamma_x },
\]
where $ \gamma_x = \arccos\mypar{\dfrac{b}{\sqrt{b^2+c^2}}} $. Then, we can rewrite  $ f_{Q_x}(0^m) $ as
\[
    a + 2 \sqrt{b^2+c^2} \cdot x^{m/2} \cdot \cos \mypar{m \theta_x+\gamma_x},
\]
where $a$, $ \sqrt{b^2+c^2} $, and $ x^{m/2} $ are all positive values.

By picking the cutpoint $ \lambda_x = a = \dfrac{1}{3x+1} $, we can have
\[
    f_{Q_x}(0^m)  = \lambda_x + 2 \sqrt{b^2+c^2} \cdot x^{m/2} \cdot \cos \mypar{m \theta_x+\gamma_x}.
\]
Thus,  $ f_{Q_x}(0^m) $ is greater than the cutpoint $ \lambda_x $ if and only if $ \cos \mypar{m \theta_x+\gamma_x}  $ is positive. Remark that since $ x \in (0,1/2] $, 
\begin{equation}
    \label{eq:intervals1}
    \theta_x \in \left( \frac{2\pi}{4}, \frac{3 \pi}{4} \right] ~ \mbox{ and } ~ \gamma_x \in \left( \frac{9\pi}{18},\frac{11\pi}{18} \right).
\end{equation}

\begin{fact} 
    \label{fact:different-sign}
    \cite{ShurY16}
    For any $ x_1 , x_2 \in (0,1/2]  $ with $ x_1 < x_2 $, the language recognized by $ Q_{x_1} $ with cutpoint $ \lambda_{x_1} $ is different than the language recognized by $ Q_{x_2} $ with cutpoint $ \lambda_{x_2} $.
\end{fact}
\begin{proof}
    We can conclude the proof by showing the existence of a string, say $ 0^n $, such that $  \cos \mypar{n \theta_{x_1}+\gamma_{x_1}} $ and  $ \cos \mypar{n \theta_{x_2}+\gamma_{x_2}} $ have different signs.
    
    By checking the values of $ \arccos(-\sqrt{x}) $ on the interval $ (0,1/2] $, we can easily see that, if $ x_1 < x_2 $, then $ \theta_{x_1} < \theta_{x_2} $. Due to Eq.~\ref{eq:intervals1}, we know that $ \theta_{x_2} - \theta_{x_1} $ is always less than $ \frac{\pi}{4} $. Besides, $ | \gamma_{x_2}-\gamma_{x_1} | < \frac{\pi}{9} $. Thus, we can say that there exists an integer $ m $ such that 
    \begin{equation}
        \label{eq:theta-gamma-diff}
        m \mypar{\theta_{x_2}-\theta_{x_1}} + \gamma_{x_2}-\gamma_{x_1} \leq \pi
        ~\mbox{ and }~
        \pi < (m+1) \mypar{\theta_{x_2}-\theta_{x_1}} + \gamma_{x_2}-\gamma_{x_1} < 2\pi.
    \end{equation}
    
    The real line can be partitioned into intervals of length $ \pi $, in which the function $ \cos(\cdot) $ does not change its sign --- all borderline points are attached to ``negative'' intervals:
    
    \centerline{
    \unitlength=1mm
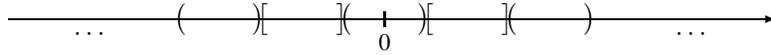
    \begin{picture}(100,12)(0,-5)
    \put(50,-2){\makebox(0,0)[ct]{\small0}}
    \put(0,0){\vector(1,0){102}}
    \put(50,-1){\line(0,1){2}}
    \multiput(34,0)(22,0){3}{\makebox(0,0)[rc]{\big)}}
    \multiput(22,0)(22,0){3}{\makebox(0,0)[lc]{\big(}}
    \multiput(33,0)(22,0){2}{\makebox(0,0)[lc]{\big[}}
    \multiput(45,0)(22,0){2}{\makebox(0,0)[rc]{\big]}}
    \multiput(11,-2)(80,0){2}{\makebox(0,0)[cc]{$\cdots$}}
    \end{picture}
    }
    
    Let $ \beta_1 = (m+1) \theta_{x_1} + \gamma_1 $ and $ \beta_2 = (m+1) \theta_{x_2} + \gamma_2 $. We know that $ \beta_2 - \beta_1 $ is greater than $ \pi $ and less than $ 2\pi $. Therefore, $ \beta_1 $ and $ \beta_2 $ lie in the consecutive intervals, and hence they have different signs. Then, $ 0^{m+1} $ is a string that separates both languages. 
\end{proof}

\section{Our construction for unary PFAs}
\label{sec:our-unary-PFA}

We present our construction for 3-state unary PFAs that can define uncountably many languages with a fixed cutpoint (e.g., $ \frac{1}{2} $). Similarly to binary PFA case, we introduce another parameter $\alpha$ in matrix $ B_x $. Besides, we restrict the interval of $ x $ with $ (0,1/10) $. Let $ \alpha \in (1/2,1] $. The new matrix is defined as 
\[
    B_{x,\alpha} = \mymatrix{ccc}{ 1-\alpha & 1-\alpha &1+x-\alpha  \\ \alpha & \alpha-1 & \alpha+x-1\\
    0 & 1 & 1-2x }.
\]
Even though $ B_{x,\alpha} $ is not a stochastic matrix in general, its eigenvalues are identical to the eigenvalues of $ B_x $: 
\[
    r_1' = 1, ~~~ r_2' = -x + \sqrt{x-x^2} \cdot i, ~\mbox{ and }~ r_3'=-x - \sqrt{x-x^2} \cdot i.
\]
Similarly to the calculations of $ B_x $, we can write $ B_{x,\alpha}^m(3,1) $ as
\[
    a'+(b'+c'\cdot i)r_2^m+(b'-c'\cdot i)r_3^m
\]
for real values $a'$, $b'$, and $c'$. By using the initial conditions
\[
    B_{x,\alpha}^0(3,1)=0, ~~~
    B_{x,\alpha}^1(3,1)=0, ~\mbox{ and }~
    B_{x,\alpha}^2(3,1)=\alpha,
\]
the coefficients can easily be found as 
\[
     a' = \dfrac{\alpha}{3x+1}, ~~~ b' = -\dfrac{\alpha}{6x+2}, ~ \mbox{ and } ~ c' = \alpha \dfrac{x+1}{(6x+2)\sqrt{x-x^2}}.
\]
In other words, $ B_{x,\alpha}^m(3,1) = \alpha B_{x}^m(3,1) $, or equivalently 
\begin{equation}
    \label{eq:Bxalpha}
    B_{x,\alpha}^m(3,1) = \frac{\alpha}{3x+1} + 2 \alpha \sqrt{b^2+c^2}\cdot x^{m/2} \cdot \cos(m\theta_x + \gamma_x ).
\end{equation}
This time, since we restrict $ x \in (0,1/10] $, we have
\begin{equation}
    \label{eq:thetagammainterval-alpha}
    \theta_x \in \left( \frac{9\pi}{18},\frac{11\pi}{18} \right)
    ~\mbox{ and }~
    \gamma_x \in \left( \frac{9\pi}{18},\frac{11\pi}{18} \right).
\end{equation}

Surprisingly, $ B^3_{x,\alpha} $ is a stochastic matrix:
\[
    B'_{x,\alpha} = B^3_{x,\alpha} = \mymatrix{ccc}{
         1-\alpha + \alpha x      & 1-\alpha +\alpha x - 2x^2 &1+\alpha x-\alpha-3x^2+4x^3  \\
        \alpha x    & \alpha x +x-2x^2& x+\alpha x-5x^2+4x^3\\
        \alpha-2 \alpha x  & \alpha-2\alpha x-x+4x^2 & \alpha- 2\alpha x- x+8x^2-8x^3
    }.
\]
Thus, we can define our new unary 3-state PFA, say $Q_{x,\alpha}$, by using $  B'_{x,\alpha} $ as the single transition matrix. The initial state is the first state, and the only accepting state is the third state. Therefore, by Eq.~\ref{eq:Bxalpha},
\[
    f_{Q_{x,\alpha}}(0^m) = \alpha f_{Q_{x}}(0^{3m}) = 
    \frac{\alpha}{3x+1} + 2 \alpha \sqrt{b^2+c^2}\cdot x^{3m/2} \cdot \cos(3m\theta_x + \gamma_x ) ,
\]
where $ m \geq 0 $. By picking $ \alpha = \frac{3x+1}{2} $, we can get $ \frac{\alpha}{3x+1} = \frac{1}{2} $. Let $ Q'_{x} $ be the PFA $ Q_{x,\alpha} $ where $ \alpha = \frac{3x+1}{2} $. Then, we can say that $ f_{Q'_{x}}(0^m) $ is greater than $ \frac{1}{2} $ if and only if $ \cos(3m\theta_x+\gamma_x) >0$.

\begin{theorem}
    For any $ x_1,x_2 \in (0,1/10) $ with $ x_1 < x_2 $, the language recognized by $ Q'_{x_1} $ with cutpoint $\frac{1}{2}$ is different than the language recognized by $ Q'_{x_2} $ with cutpoint $ \frac{1}{2} $.
\end{theorem}
\begin{proof}
    Due to Eq.~\ref{eq:thetagammainterval-alpha}, we know that $ \theta_{x_2} - \theta_{x_1} $ is always less than $ \frac{\pi}{9} $, and hence $ 3(\theta_{x_2} - \theta_{x_1}) $ is less than $ \frac{\pi}{3} $. Moreover, $ 3(\theta_{x_2} - \theta_{x_1}) + \gamma_2 - \gamma_1 $ is less than $ \frac{4\pi}{9} $. Thus, we can say that there exists an integer $ m $ such that
    \begin{equation}
        \label{eq:theta-gamma-diff-alpha}
        3m \mypar{\theta_{x_2}-\theta_{x_1}} + \gamma_{x_2}-\gamma_{x_1} \leq \pi
        ~\mbox{ and }~
        \pi < 3(m+1) \mypar{\theta_{x_2}-\theta_{x_1}} + \gamma_{x_2}-\gamma_{x_1} < 2\pi.
    \end{equation}
    
    Let $ \beta'_1 = 3(m+1) \theta_{x_1} + \gamma_1 $ and $ \beta'_2 = 3(m+1) \theta_{x_2} + \gamma_2 $. We know that $ \beta'_2 - \beta'_1 $ is greater than $ \pi $ and less than $ 2\pi $. Therefore, as explained in Fact~\ref{fact:different-sign}, $ \cos(\beta'_1) $ and $ \cos(\beta'_2) $ have different signs. Then, $ 0^{m+1} $ is the string that separates both languages. 
\end{proof}

\begin{corollary}
    The languages recognized by the family of 3-state unary PFAs $ \{ Q'_{x} \mid x \in (0,1/10) \} $ with cutpoint $ \frac{1}{2} $ form an uncountable set.
\end{corollary}

\acknowledgements
\label{sec:ack}
We thank two anonymous reviewers for their comments and corrections.

%\nocite{*}
%\bibliographystyle{abbrvnat}
% use the following instead if you encounter problems 
\bibliographystyle{alpha}
\bibliography{ref}
\label{sec:biblio}

\end{document}